\documentclass[11pt]{article}
\usepackage{fullpage}
\usepackage{hyperref}
\usepackage{url}
\usepackage{graphicx}
\usepackage{amssymb}
\usepackage{amsmath}
\usepackage{amsthm}
\usepackage{epsfig}
\usepackage{color}
\usepackage{algorithm}
\usepackage[noend]{algorithmic}

\newtheorem{theorem}{Theorem}

\newtheorem{lemma}[theorem]{Lemma}
\newtheorem{cor}[theorem]{Corollary}

\newtheorem{obs}[theorem]{Observation}

\newcommand{\eps}{\varepsilon}
\newcommand{\opt}{\mbox{opt}}
\newcommand{\OPT}{\mbox{\sc OPT}}
\newcommand{\prob}{\mbox{{\bf Pr}}}
\newcommand{\expect}{\mbox{{\bf E}}}

\newcommand{\tuple}[1]{\left(#1\right)}

\newcommand{\wt}[1]{\mbox{\textbf{wt}}\tuple{#1}}

\def\poly{{\rm poly}}
\def\eps{\varepsilon}

\title{A Weakly Robust PTAS for Minimum Clique Partition\\ in Unit Disk Graphs}

\author{Imran A. Pirwani\thanks{Department of Computing Science,
University of Alberta, 
Edmonton, Alberta T6G 2E8, Canada.
Email: \texttt{pirwani@cs.ualberta.ca}. Supported by Alberta Ingenuity.
}
\and
Mohammad R. Salavatipour\thanks{Department of Computing Science,
University of Alberta, 
Edmonton, Alberta T6G 2E8, Canada.
Email: \texttt{mreza@cs.ualberta.ca}. Supported by NSERC and Alberta Ingenuity.
}
}

\begin{document}
\date{}
\bibliographystyle{plain}

\maketitle





\begin{abstract}
We consider the problem of partitioning the set of vertices of a
given unit disk graph (UDG) into a minimum number of cliques. 
The problem is NP-hard and various constant factor approximations are 
known, with the current best ratio of $3$.
Our main result is a {\em weakly robust} polynomial time approximation 
scheme (PTAS) for UDGs expressed with edge-lengths, it
either (i) computes a clique partition or (ii) gives a certificate that 
the graph is not a UDG;
for the case (i) that it computes a clique partition, we show that 
it is guaranteed to be within $(1+\eps)$ ratio of the optimum if the 
input is UDG;
however if the input is not a UDG it either computes a clique partition 
as in case (i)
with no guarantee on the quality of the clique partition or detects that 
it is not a UDG.  Noting that 
recognition of UDG's is NP-hard even if we are given edge lengths, 
our PTAS is a weakly-robust algorithm. 
Our algorithm can be transformed into an 
$O\left(\frac{\log^* n}{\eps^{O(1)}}\right)$ 
time distributed PTAS.

We consider a weighted version of the clique partition
problem on vertex weighted UDGs that
generalizes the problem.  We note some
key distinctions with the unweighted version,
where ideas useful in obtaining a PTAS breakdown.
Yet, surprisingly, it admits a 
$(2+\eps)$-approximation algorithm for the 
weighted case where the graph is expressed,
say, as an adjacency matrix. This improves on the best known 
$8$-approximation for the {\em unweighted} case for UDGs expressed in 
standard form.

\end{abstract}



\section{Introduction} \label{sec:intro}
A standard network model for homogeneous networks
is the {\em unit disk graph} (UDG). 
A graph $G=(V,E)$ is a UDG if there is a
mapping $f: V \mapsto \mathbb{R}^2$ such that $\|f(u)-f(v)\|_2 \leq 1
\Leftrightarrow \{u,v\} \in E$;  $f(u)$\footnote{$f(.)$ is called
a realization of $G$. Note that $G$ may not come with
a realization.} models the position of the
node $u$ while the unit disk centered at $f(u)$ models the range of
radio communication.  Two nodes $u$ and $v$ are said to be able to
directly communicate if they lie in the unit disks placed at each
others' centers.  There is a vast collection of literature on
algorithmic problems studied on UDGs. See the survey \cite{BalasundaramB09}.

Clustering of a set of points is an important 
subroutine in many algorithmic and practical applications and
there are many kinds of clusterings depending upon the application.  
A typical objective in clustering is to minimize the 
number of ``groups" such that each ``group" (cluster) satisfies a set 
of criteria.  Mutual proximity of points in a cluster is
one such criterion, while points in a cluster forming a clique in
the underlying network is an extreme form of mutual proximity.
We study an optimization problem related to clustering, called
the {\em minimum clique partition} problem on this UDGs.  

{\bf Minimum clique partition on unit disk graphs (MCP):}
Given a unit disk graph, $G=(V,E)$, partition $V$ into a smallest
number of cliques.  

Despite being theoretically interesting, MCP has
been useful for other problems.  For example, 
\cite{PemmarajuPMobihoc06} shows how to use a 
small-sized clique
partition of a UDG to construct a large collection of disjoint
(almost) dominating sets.  They \cite{PemmarajuPESA07} also show
how to obtain a good quality realization of UDGs, and an 
important ingredient in their technique was to construct a small-sized 
clique partition of the graph.  It is shown \cite{LillisPPADHOCNOW07}
how to use a small-sized clique partition to obtain 
sparse spanners with bounded dilation, which
also permit guaranteed geographic routing on a related class
of graphs. \cite{PanditP09} 
employ MCP to obtain an $O(\log^* n)$ time distributed algorithm which 
is an $O(\log n)$-approximation for the facility location problem
on UDGs without geometry; they also give an $O(1)$ time distributed
$O(1)$-approximation to the facility location problem on UDGs with
geometry also using MCP. Recently, \cite{PanditPV09}
shows how to obtain a first $O(1)$ approximation to the domatic
partition problem on UDGs using MCP.

On general graphs, the clique-partition problem is equivalent to the 
minimum graph coloring on the complement graph which is 
not approximable within $n^{1-\eps}$, for any $\eps > 0$, unless P=NP
\cite{Zuckerman}. 
MCP has been studied for special graph classes. 
It is shown to be MaxSNP-hard for 
cubic graphs and NP-complete for planar cubic graphs \cite{CerioliFFMPR08};  
they also give a 
$5/4$-approximation algorithm for graphs with maximum degree at most $3$.
MCP is NP-hard for a 
subclass of UDGs, called {\em unit coin graphs}, where the interiors of
the associated disks are pairwise disjoint \cite{CerioliFFP04}.  
Good approximations, however, are possible on UDGs.  
The best known approximation is due to 
\cite{CerioliFFP04} who give a $3$-approximation  
via a partitioning the vertices into co-comparability graphs,
and solving the problem exactly on them.
They give a $2$-approximation algorithm for coin graphs.
MCP has also been studied on UDGs expressed in standard form.  
For UDGs expressed in general form 
\cite{PemmarajuPESA07} give an $8$-approximation algorithm.

\subsubsection*{Our Results and Techniques:}
In this paper we present a weakly-robust\footnote{An 
algorithm is called {\em robust} if it either computes an answer or
declares that the input is not from the restricted domain; if the
algorithm computes an answer then it is correct \cite{RaghavanS03}. 
We call our algorithm
{\em weakly-robust} in that it always computes a clique partition or declares that
the input is not from the restricted domain (i.e. not a UDG); if the input happens to
be a UDG then the answer is a $(1+\epsilon)$-approximate clique partition. Otherwise,
it still returns a clique partition but there is no guarantee on the quality of the clique partition.} 
PTAS for MCP on a given UDG.
For ease of exposition, first we prove this (in Section~\ref{sec:2a}) 
when the UDG is given with a realization, $f(.)$.
The holy-grail is a PTAS when the UDG is expressed in
standard form, say, as an adjacency matrix. 
However, falling short of proving this, we show
(in Section~\ref{sec:2b}) how to get a PTAS when 
the input UDG is expressed in standard form along with associated 
edge-lengths corresponding to some (unknown) realization.  
The algorithm is weakly-robust in the sense that
it either (i) computes a clique partition of the input graph or (ii) gives a 
certificate that the input graph is not a UDG.
If the input is indeed a UDG then the algorithm returns a clique partition
(case (i)) which is a 
$(1+\eps)$-approximation (for a given $\epsilon>0$). 
However, if the input is not a UDG, the algorithm either computes
a clique partition but with no guarantee on the quality of the solution
or returns that it is not a UDG.
Therefore, this algorithm should be seen as a weakly-robust PTAS.
The generation of a polynomial-sized certificate
which proves why the input graph is not a UDG should be seen in the
context of the negative result of \cite{AspnesGY04}
which says that even if edge lengths are given, UDG recognition is NP-hard.
We show (in Section~\ref{sec:distrib}) how this algorithm can be 
modified to run in $O(\frac{\log^* n}{\eps^{O(1)}})$ 
distributed rounds.

In Section~\ref{sec:wtd} we explore a weighted version of 
MCP where we are given
a vertex weighted UDG. In this formulation, the weight of a clique
is the weight of a heaviest vertex in it, and the weight of a clique
partition is the sum of the weights of the cliques in it.  We note
some key distinctions between the weighted and the unweighted versions
of the problem and show that the ideas that help in obtaining a PTAS
do not help in the weighted case.  Yet, surprisingly, we show that
the problem admits a $(2+\eps)$-approximation algorithm for the 
weighted case {\em using only
adjacency}.  This result should be contrasted with the unweighted
case where it is not clear as to how to remove the dependence on
the use of edge-lengths, which was crucially
exploited in deriving a PTAS.

We use $\OPT$ to denote an optimum clique partition and
$\opt$ to denote the size (or, in Section~\ref{sec:wtd}, weight) of an 
optimum clique partition. We also use $n$ and $m$ to denote the number
of points (i.e. nodes of $G=(V,E)$) and the number of edges, respectively.

\section{A Weakly-Robust PTAS for UDG Expressed with Edge-lengths}  
\label{sec:2b}

For simplicity, we first describe an
algorithm when the input is given with a geometric realization.

\subsection{A PTAS for UDGs With a Geometric Realization} 
\label{sec:2a}

We assume the input UDG is expressed 
with geometry of its points.  Using a randomly shifted grid whose cell size
is $k\times k$ (for $k=k(\eps)$) we partition the
plane. Since the diameter of the convex hull of each clique is at most
1, for large values of $k$, a fixed clique 
is cut by this grid (and therefore belongs to at most four cells) with
probability at most $\frac{2}{k}$.  Therefore, if we could efficiently compute an optimal clique 
partition in each $k \times k$ cell, then taking the union of these
cliques yields a solution whose expected size at most $(1+\eps)
\opt$. We can easily repeat this process $O(\log n )$ times to obtain a solution with
size at most $(1+\eps)\opt$ w.h.p. We call the algorithm {\tt MinCP1}, formalized below.

\begin{theorem} \label{T2}
Algorithm MinCP1 (given below) returns, in poly-time, a clique partition of size at 
most $(1+\eps)\opt$ w.h.p. 
\end{theorem}
\begin{algorithm*}[ht]
	\caption{MinCP1$(G,\eps)$}
	\begin{algorithmic}[1]
	\STATE Let $k = \lceil\frac{16}{\eps}\rceil$.  Place a grid whose
				 squares have size $k \times k$, on the plane.  
				 Call it $\mathcal{G}_{0,0}$.	
	\STATE Pick $(a,b) \in [0, k) \times [0,k)$ uniformly at random.
	\STATE Shift $\mathcal{G}_{0,0}$ by $(a,b)$ to get $\mathcal{G}_{a,b}$
				 which is a grid shifted $a$ units to the right and $b$ units
				 above.  $\mathcal{G}_{a,b}$ induces a random partition of
				 $V$ into points in $k \times k$ regions. 
	\FORALL {$k \times k$ regions of $\mathcal{G}_{a,b}$}
		\STATE Obtain an optimal partition $C_i$ 
					 for point-set $P$ in the $k \times k$ square.
	\ENDFOR
	\STATE Let $\mathcal{C}_{a,b} = \bigcup_{i=1}^t C_i$ be the
				 union of clique partitions obtained for
				 the points in each $k \times k$ square.
	\STATE Repeat ``Step 2--6" $\lceil\log n\rceil$ times and return
				 the smallest $\mathcal{C}_{a,b}$ over the $\lceil\log n\rceil$
				 independent trials.
	\end{algorithmic}
\end{algorithm*}

We begin with a simple observation.
\begin{obs}
The diameter of the convex hull of every clique is at most $1$.
\end{obs}
In the next subsection we argue how to perform ``Step 5" of the
algorithm MinCP1 efficiently. Assuming this, we prove Theorem~\ref{T2}.  
For a random shift $\mathcal{G}_{a,b}$ and a clique $C$, we say that 
$\mathcal{G}_{a,b}$ ``cuts" $C$ if some line of $\mathcal{G}_{a,b}$
crosses an edge of $C$.  It is easy to see that:

\[\prob	\left[\substack{C \text{ is cut} \\
	\text{ by } \mathcal{G}_{a,b}}\right]  \leq 
				\prob\left[\substack{\text{a vertical or horizontal line of } \mathcal{G}_{a,b} \\
	\text{crosses an edge of } C}\right] \leq  \frac{2}{k}
\]

Thus, the expected number of cliques in an optimal partition that are ``cut" by $\mathcal{G}_{a,b}$ is at most 
$\frac{2}{k} \cdotp \opt$. So, by Markov's inequality, with probability at least $1/2$ there are no more than
$\frac{4}{k}\cdot\opt$ cliques cut by $\mathcal{G}_{a,b}$.
Therefore, if we compute an optimal solution for each of the $k\times k$ grid cells and take the union of them,
with probability at least $1/2$ we get an excess of at most $4\times\frac{4}{k}\cdotp \opt $ cliques with respect to optimum
since each clique that is ``cut'' by the grid can be counted up to four times. If we repeat this process 
for $\lceil\log n\rceil$ independent random trials, we get that 
with probability at least $1-\frac{1}{n}$ the size of the solution we obtain is at most
$\opt + \frac{16}{k} \cdotp \opt \leq  (1 + \eps) \cdotp \opt$.


\subsubsection{Optimal Clique Partition of a UDG in a $k \times k$ Square}
\label{exactgeom}
Unlike optimization problems such as  maximum
(weighted) independent set and minimum dominating set, where 
one can ``guess" only a small-sized subset of
points to obtain an optimal solution, the
combinatorial complexity of any single clique in an optimal solution
can be high. Therefore, it is unclear as to how to ``guess" 
even few cliques, each of which may be large.
A result of Capoyleas et al. \cite{CapoyleasRW91} comes to our aid;
a version of their result says that there exists an optimal clique
partition where the convex hulls of the cliques are pair-wise
non-overlapping.
This phenomenon of {\em separability} of an optimal partition, coupled
with the fact that the size of an optimal partition in a 
small region is small, allows us to circumvent the above difficulty.
The following simple lemma bounds the size 
of an optimal solution of an instance of bounded diameter.

\begin{lemma}
\label{lemma:fewcliques}
Any set of points $P$ in a $k \times k$ square has a clique partition 
of size $O(k^2)$.  
\end{lemma}
\begin{proof}
Place a grid whose cells have size $1/2 \times 1/2$.  This grid
induces a vertex partition where each block in the partition consists of the
points that share a common grid cell (and therefore form a clique).  
\end{proof}

We state a variant of a result by Capoyleas et al.
\cite{CapoyleasRW91} according to which there exists an
optimal clique partition where the convex hulls of the cliques
are non-overlapping, that is, for any pair of cliques in an optimal
partition, there is a straight line which separates them.\footnote{
We gave a proof of this theorem \cite{arxiv} before it was brought
to our attention that Capoyleas, Rote, and Woeginger \cite{CapoyleasRW91}
proved this much earlier in a different context.}

\begin{figure}[h]
\centering
\begin{tabular}{c@{\hspace{0.1\linewidth}}c}
\includegraphics[height = 0.25\linewidth]{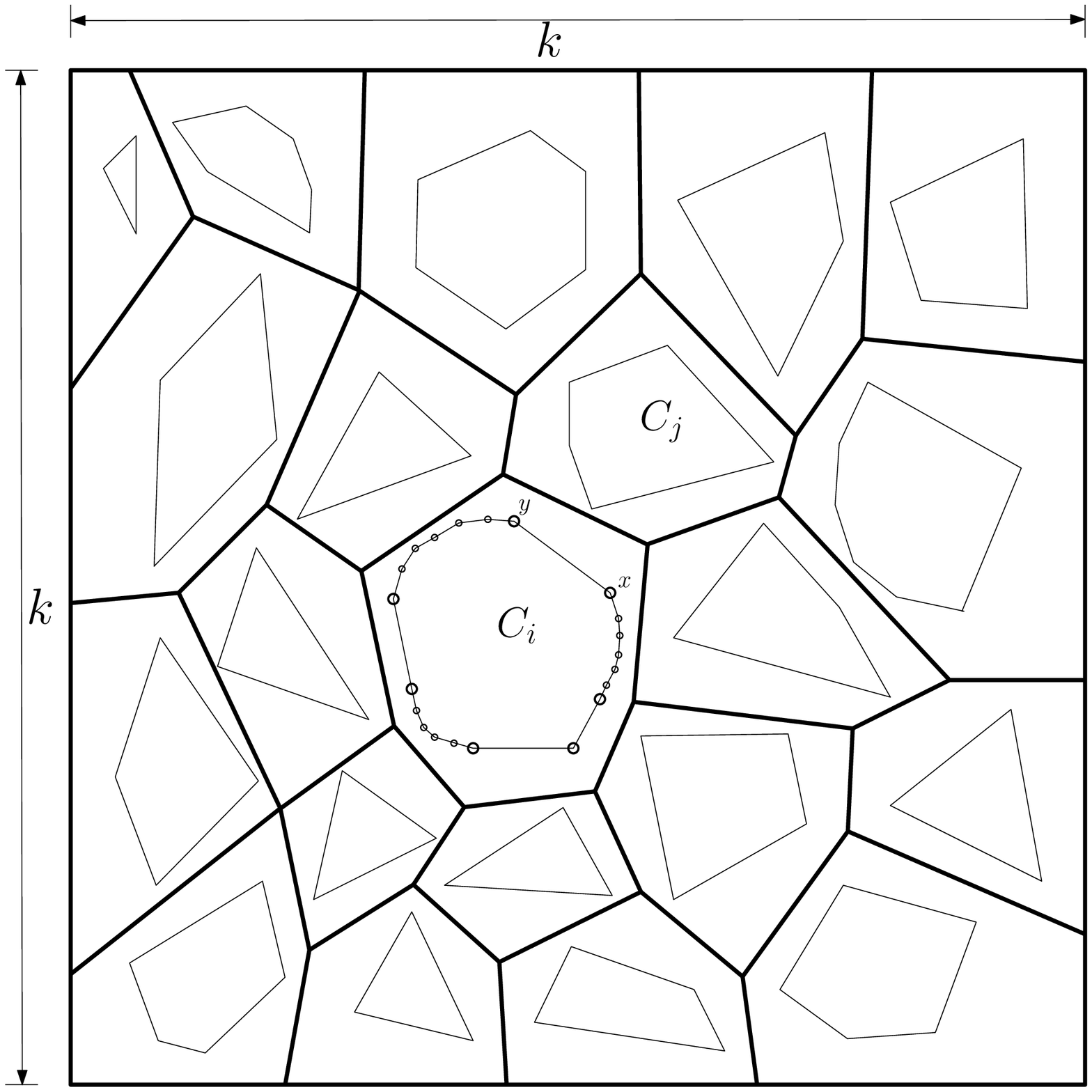} &
\includegraphics[height = 0.25\linewidth]{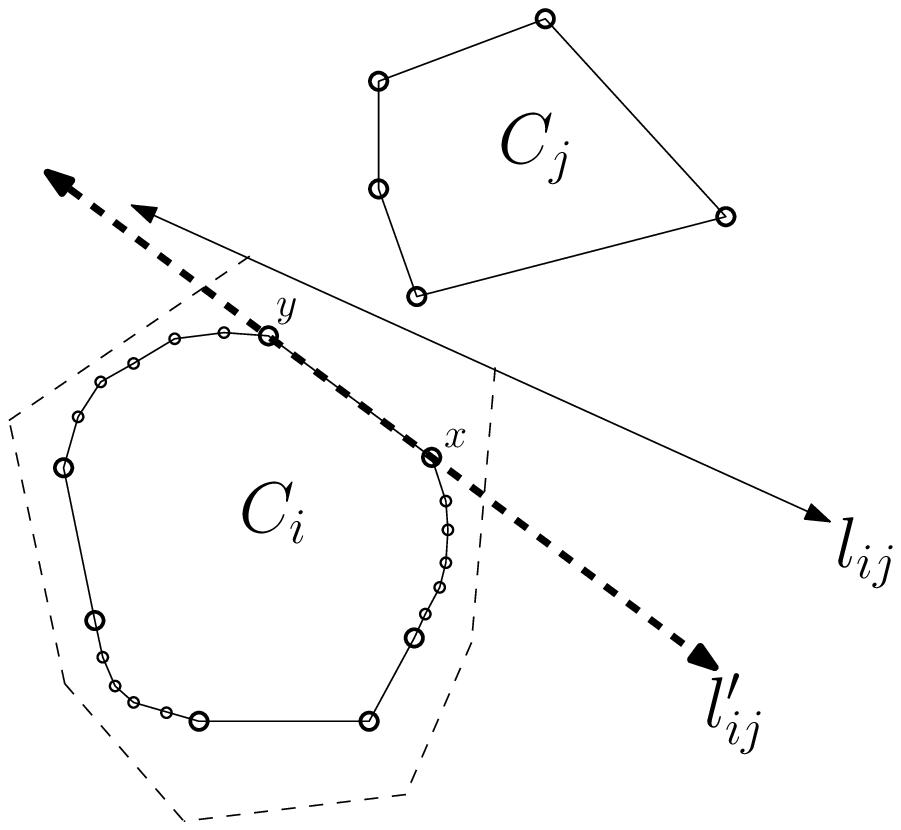} \\
(a) & (b)
\end{tabular}
\caption{(a) An optimal clique partition of UDG points in a
bounded region; each light convex shape corresponds to a clique
in the clique partition.  The heavy
line-segments represent segments of the corresponding separators. 
(b) A close-up view of $C_i$ and $C_j$.  A separator line, $l_{ij}$
is shown which separates $C_i$ and $C_j$, corresponding to the segment
in (a). Note that $l_{ij}^\prime$ is also a separator for $C_i$ and
$C_j$ and $l_{ij}^\prime$ is passing through points $x$ and $y$ in
$C_i$. }
\label{figure:convexregions}
\end{figure}

\begin{theorem}[\cite{CapoyleasRW91}]
\label{theorem:cliqueseparable}
\label{separationthm}
For a clique partition in which the convex hulls of the cliques
are pairwise non-overlapping, 
there is a straight line $l_{ij}$ that separates a pair of cliques
$C_i$, $C_j$ such that all vertices of $C_i$
are on one side of $l_{ij}$, and all the vertices of $C_j$ are on the 
other side of $l_{ij}$. 
(see Figure~\ref{figure:convexregions}). Furthermore, this partition can
be computed in poly-time.
\end{theorem}

The general structure of the algorithm for computing optimal solution of a $k\times k$ cell
is as follows. 
In order to reduce the search space for separator lines, 
one can find a characterization of the separator lines with some extra properties. 
Let $C_i, C_j$ be a pair of cliques each having at least two points. 
Let $L_{ij}$ be the (infinite) set
of distinct separator lines.  Since $C_i$ and $C_j$ are convex, there
exists at least one line in $L_{ij}$ that goes through two points
of $C_i$ (or $C_j$) (see Figure ~\ref{figure:convexregions}(b)). 
Therefore, given two cliques $C_i$ and $C_j$ in a clique partition (with pairwise non-overlapping parts) there is a separator line
$l_{ij}$ that goes through two vertices of one of them, say $u,v\in C_i$ such that all the vertices
of $C_j$ are on one side of this line and all the vertices of $C_i$ are on the other side or on the line.
Since there are $O(k^2)$ cliques in an optimal partition of $k\times k$ cell, there are $O(k^4)$ 
pairs of cliques in the partition and their convex hulls are pairwise 
non-overlapping. In fact, a more careful analysis shows that the dual graph of the regions is planar (see
Figure \ref{figure:convexregions}(a));
thus there are $O(k^2)$ distinct straight lines, 
each of which separate a pair of cliques in our optimal 
solution.  For every clique $C_i$, the separator
lines $l_{ij}$ (for all values of $j$) define a convex region that contains
clique $C_i$. So once we guess this set of $O(k^2)$ lines, these 
convex regions define the cliques. 
We will try all possible (non-equivalent) sets of $O(k^2)$
separator lines and check if each of the convex regions indeed defines a 
clique and if we obtain a clique partition. 
This can be performed in $O(n^{k^2})$ time (see \cite{CapoyleasRW91} for more details).


\subsection{A PTAS for UDGs With Edge-Lengths Only}
We weaken our assumption on having access to geometry; 
we assume only edge-lengths are known
with respect to a feasible (unknown) realization of the UDG.
We prove that,
\begin{theorem} \label{T4}
Given a graph $G$ with associated (rational) edge-lengths and $\eps > 0$, 
there is a polynomial time algorithm which either computes a clique partition of $G$ 
or gives a certificate that $G$ is not a UDG. If $G$ is a UDG, the size of the clique partition
computed is a $(1+\eps)$-approximation of the optimum clique partition 
(but there is no guarantee on the size of the clique partition if the input graph is not UDG).
\end{theorem}
The high level idea of the algorithm is as follows.
As in the geometric case, we first decompose the graph into 
bounded diameter regions
and show that if we can compute the optimum clique partition of each 
region then the union
of these clique partitions is within $(1+\eps)$ fraction of the optimum.
There are two main difficulties here for which we need new ideas.
The first major difference is that, we cannot use the random shift
argument as in the geometric case.
To overcome this, we use a ball growing technique
that yields bounded diameter regions. This is inspired by
\cite{NiebergHK08} who give
local PTAS for weighted independent set, and minimum dominating
set for UDGs without geometry.  
The second major difference is that, even if we have the set of points belonging to a bounded
region (a ball) it is unclear as to how to use the separation theorem to obtain an optimal solution
for this instance. Note that we are not guaranteed to have
a UDG as input. We show that we can either compute a clique partition for each subgraph induced by a ball,
or give a certificate that the subgraph is not UDG.
If it is a UDG, then our clique partition is optimal but if it
is not a UDG there is no guarantee on its size.

Let $B_r(v) = \{u : d(u,v) \leq r\}$, where by $d(u,v)$ we mean the number 
of edges on a shortest path from $u$ to $v$.
So, $B_r(v)$ can be computed using
a breadth-first search (BFS) tree rooted at $v$.
We describe our decomposition algorithm which partitions 
the graph into bounded diameter subgraphs in Algorithm 2.
We will describe a procedure, called \mbox{OPT-CP} which, given a 
graph induced by the vertices of $B_r(v)$ 
and a parameter $\ell=\poly(r)$, runs in time 
$|B_r(v)|^{O(\ell^2)} \leq n^{O(\ell^2)}$ and 
either produces a certificate that $B_r(v)$ is not a
UDG or computes a clique partition of $B_r(v)$; this clique partition 
is optimum if $B_r(v)$ is a UDG. 
We only call this procedure for ``small" values of $r$.

\begin{algorithm*}[h]
	\caption{MinCP2$(G,\eps)$}
	\begin{algorithmic}[1]
		\STATE {$\mathcal{C} \leftarrow \emptyset$; 
						$\beta \leftarrow \lceil c_0
						\frac{1}{\eps}\log\frac{1}{\eps} \rceil$;
						$\ell \leftarrow c_1 \beta^2$.}

		\COMMENT {where $c_0$ is the constant in Lemma~\ref{lemma:diambound},
						and $c_1$ is the constant in inequality (\ref{c1-const}).}
		\WHILE {$V \neq \emptyset$}
			\STATE Pick an arbitrary vertex $v\in V$
			\STATE $r \leftarrow 0$

			\COMMENT {Let $C_{r}(v)$ denote a clique partition 
								of $B_r(v)$ computed by calling \mbox{OPT-CP}}
			\WHILE {$|C_{r+2}(v)| > (1+\eps) \cdotp |C_{r}(v)|$}
				\STATE $r \leftarrow r + 1$
				\IF {($r > \beta$) or (\mbox{OPT-CP}$(B_r(v))$ returns 
						``not a UDG'')}
					\RETURN {``$G$ is not a UDG" and produce $B_r(v)$ as the
									 certificate}
				\ENDIF
			\ENDWHILE
			\STATE {$\mathcal{C} \leftarrow \mathcal{C} \cup C_{r+2}(v)$}
			\STATE $V \leftarrow V \setminus B_{r+2}(v)$
		\ENDWHILE
	\RETURN {$\mathcal{C}$ as our clique partition}
	\end{algorithmic}
\end{algorithm*}

Clearly, if the algorithm returns $\mathcal{C}$ on
``Step 11", it is a clique partition.
Let us assume that each ball $B_r(v)$ we consider induces a UDG and 
that the procedure \mbox{OPT-CP}
returns an optimal clique partition $C_r(v)$ for ball $B_r(v)$.
We show that in this case $|\mathcal{C}|\leq(1+\eps)\opt$.
We also show that for any iteration of the outer ``while--loop",
``Step 5" of MinCP2 is executed in time polynomial in $n$, by
using edge-lengths instead of Euclidean coordinates. 

For an iteration $i$ of the outer loop, 
let $v_i$ be the vertex chosen 
in ``Step 3'' and let $r^*_i$ be the value
of $r$ for which the ``while-loop" on ``Step 5"
terminates, that is, $|C_{r^*_i+2}(v_i)|\leq (1+\eps)\cdot|C_{r^*_i}(v_i)|$.
Let $k$ be the maximum number of iterations of the outer loop.
The following lemmas show that two distinct balls grown around vertices
are far from each other, that the union of the optimal solutions to
the balls form a lower-bound on the cost of the entire instance, and
that the cost of $\mathcal{C}$ and $\opt$ is within
a factor $(1+\eps)$ of $\opt$. 

\begin{lemma}
For every $i\not=j$, every pair $v\in B_{r^*_i}(v_i)$ and $u\in B_{r^*_j}(v_j)$
are non-adjacent.
\label{lemma:disjointballs}
\end{lemma}
\begin{proof}
Without loss of generality, let  $i<j$.
Therefore, every vertex in $B_{r^*_j}(v_j)$ is at a level larger than 
$r^*_i+2$ of the BFS tree rooted at $v_i$,
otherwise it would have been part of the ball $B_{r^*_i+2}(v_i)$ thus 
removed from $V$.
Note that in a BFS tree rooted at $v_i$, there cannot be an edge 
between a level $r$ and $r'$ with $r'\geq r+2$.
Thus there cannot be an edge between a node in $v\in B_{r^*_i}(v_i)$, 
which has level at most $r^*_i$ and
a node $u\in B_{r^*_j}(v_j)$, which would been at a level at least 
$r^*_i+3$ in the BFS tree rooted at $v_i$.
\end{proof}

Next, we derive a lower-bound on $\opt$. 

\begin{lemma}
$\opt \geq \displaystyle\sum_{i=1}^k |C_{r^*_i}(v_i)|$
\label{lemma:optlowerbound}
\end{lemma}
\begin{proof}
Note that $B_{r^*_i}(v_i)$ is obtained by constructing a BFS tree rooted at
vertex $v_i$ up to some depth $r^*_i$.  
According to Lemma~\ref{lemma:disjointballs}, there is no edge between any two nodes $v\in B_{r^*_i}(v_i)$ and $u\in B_{r^*_j}(v_j)$.
So, no single clique in an optimum solution
can contain vertices from distinct $B_{r^*_i}(v_i)$ and $B_{r^*_j}(v_j)$. 
Consider the subset of cliques in an optimal clique partition of $G$ that
intersect $B_{r^*_i}(v_i)$ and call this subset $\OPT_i$.  The argument above shows that $\OPT_i$ is disjoint from $\OPT_j$.
Also, each $\OPT_i$ contains all the vertices in $B_{r^*_i}(v_i)$.
Since $C_{r^*_i}(v_i)$ is an optimal clique partition for $B_{r^*_i}(v_i)$,
$|\OPT_i| \geq |C_{r^*_i}(v_i)|$. The lemma immediately follows by observing that $\OPT_i$ and $\OPT_j$ are disjoint.
\end{proof}

The next lemma relates the cost of our solution to $\opt$.

\begin{lemma}
If $|C_{r^*_i+2}(v_i)| \leq (1+\eps) \cdotp |C_{r^*_i}(v_i)|$, then
$
	\left|\displaystyle\bigcup_{i=1}^{k} C_{r^*_i+2}(v_i)\right| \leq 
			(1+\eps) \cdotp \opt
$
\label{lemma:ptas}
\end{lemma}
\begin{proof}
$	
	\left|\displaystyle\bigcup_{i=1}^{k} C_{r^*_i+2}(v_i)\right| = 
		\displaystyle\sum_{i=1}^{k} \left|C_{r^*_i+2}(v_i)\right| \leq 
		(1+\eps) \cdotp \displaystyle\sum_{i=1}^{k} \left|C_{r^*_i}(v_i)\right| \leq 
		(1+\eps) \cdotp \opt
$
\end{proof}

Finally, we show that the inner ``while-loop" terminates in
$\tilde{O}(\frac{1}{\eps})$, so $r^*_i\in \tilde{O}(\frac{1}{\eps})$.
Obviously, the ``while-loop" on ``Step 5" terminates eventually, so $r^*_i$ exists.
By definition of $r^*_i$, for all smaller values of
$r < r^*_i$: $|C_r(v_i)| > (1+\eps) \cdotp |C_{r-2}(v_i)|$. Since diameter of $B_r(v_i)$
is $O(r)$, if $B_r(v)$ is a UDG, there is a realization of it in which all the points
fit into a $r\times r$ grid. Thus, $|C_r(v_i)|\in O(r^2)$.
So for some $\alpha \in O(1)$:

\[
	\alpha \cdotp r^2 >
	|C_r(v_i)| > (1+\eps) \cdotp |C_{r-2}(v_i)| > \ldots >
							(1+\eps)^{\frac{r}{2}} \cdotp |C_0(v_i)| = 
								O(\left(\sqrt{1+\eps}\right)^r),
\]

when $r$ is even (for odd values of $r$ we obtain $|C_r(v_i)|>(1+\eps)^{\frac{r-1}{2}} \cdotp |C_1(v_i)| \geq O(\left(\sqrt{1+\eps}\right)^{r-1})$.
Therefore we have:

\begin{lemma}\label{lemma:diambound}
There is a constant $c_0>0$ such that for each $i$: $r^*_i \leq c_0/\eps \cdotp \log{1/\eps}$.
\end{lemma}

In the next subsection, we show that the algorithm \mbox{OPT-CP}, 
given $B_{r}(v)$ and an upper bound $\ell$ on $|C_r(v)|$, 
either computes a clique partition 
or declares that the graph is not UDG; the size of the partition
is optimal if $B_r(v)$ is a UDG. The algorithm runs
in time $n^{O(\ell^2)}$.
By the above arguments, if $B_r(v)$ is a UDG then, there is a constant $c_1>0$ such that:

\begin{equation}\label{c1-const}
|C_r(v)| = O({r^*_i}^2) \leq c_1\cdot\frac{c_0^2}{\eps^2}\log^2\frac{1}{\eps}.
\end{equation}

We can 
set $\ell = \lceil c_1 \frac{c_0^2}{\eps^2}\log^2\frac{1}{\eps} \rceil$
for any invocation of OPT-CP as an upper bound,
where $c_1$ is the constant in $O({r^*_i}^2)$.
So, the running time of the algorithm is 
$n^{\tilde{O}(1/\eps^4)}$.

\subsection{An Optimal Clique Partition for $B_r(v)$}\label{secbrv}
Here we present the algorithm \mbox{OPT-CP} that given 
$B_r(v)$ (henceforth referred to as $G'$) and an upper bound $\ell$ 
on the size of an optimal solution for $G'$, either computes a 
clique partition of it or detects that it is not a UDG; if $G^\prime$
is a UDG then the partition is optimal.  The algorithm runs in time 
$n^{O(\ell^2)}$. Since, by 
Lemma \ref{lemma:diambound}, $\ell$ is a constant in each call to 
this algorithm, the running time of OPT-CP is polynomial in $n$.
Our algorithm is based on the separation theorem \cite{CapoyleasRW91}. 
Even though we do not have a 
realization of the nodes on the plane, assuming that $G'$ is a UDG, we
show how to apply the {\em separation theorem} \cite{CapoyleasRW91}
as in the geometric setting. We use node/point to refer to a vertex of $G'$ 
and/or its corresponding point on the plane for some realization of $G'$.
We will use the following technical lemma.

\begin{lemma}\label{quad-lem}
Suppose we have four mutually adjacent nodes $p,a,b,r$ and their 
pairwise distances with respect to some realization 
on the Euclidean plane. Then there is a poly-time procedure that 
can decide if $p$ and $r$ are on the same side of the line that goes 
through $a$ and $b$ or are on different sides.
\end{lemma}
\begin{proof}
First, we describe how to detect if the quadrilateral on these four points is convex or concave. 
If the quadrilateral is concave, then one of the points will be inside the triangle formed by
the other three. There are three possible cases: $r$ is inside, $p$ is inside, or one of $a$ or $b$ is
inside (see Figure \ref{quad-fig}(c)-(e)). There are four triangles each
of which is over three of these
four points. The quadrilateral is concave if the sum of the areas of three of these triangles 
is equal to the area of the fourth triangle. Equivalently, it is convex if sum of areas of
two of the triangles is equal to the sum of areas of the other two. Given a triangle with edge lengths
$x,y,z$, using Heron's formula, the area of the triangle is equal to $\sqrt{2(x^2y^2+y^2z^2+z^2x^2)-(x^4+y^4+z^4))/4}$.
So the area of a triangle is of the form $\sqrt{A}$ where $A$ is a polynomial in terms of lengths of the edges
of the triangle. Suppose that the areas of the four triangles over these four points are $\sqrt{A_1}$,
$\sqrt{A_2}$, $\sqrt{A_3}$, and $\sqrt{A_4}$. We need to check if the sum of two is equal to the sum of the
other two and we would like to do this without computing the square roots of numbers.
For instance, suppose we want to verify $\sqrt{A_1}+\sqrt{A_2}=\sqrt{A_3}+\sqrt{A_4}$. For this to hold,
we must have $A_1+A_2+2\sqrt{A_1A_2}=A_3+A_4+2\sqrt{A_3A_4}$. Verifying this is equivalent to
verifying $D+\sqrt{A_1A_2}=\sqrt{A_3A_4}$ where $D=\frac{1}{2}(A_1+A_2-A_3-A_4)$.
Taking the square of both sides, we need to have $D^2+A_1A_2+2D\sqrt{A_1A_2}=A_3A_4$, which is
the same as $\frac{1}{4}(A_3A_4-D^2-A_1A_2)^2=A_1A_2D^2$. Thus by comparing two polynomials of edge-lengths
(and without computing square roots) we can check if the quadrilateral is convex or concave.

\begin{figure}[htpb]
\centering
\begin{tabular}{c@{\hspace{0.025\linewidth}}c@{\hspace{0.025\linewidth}}c@
{\hspace{0.025\linewidth}}c@{\hspace{0.025\linewidth}}c}
\includegraphics[height = 0.15\linewidth]{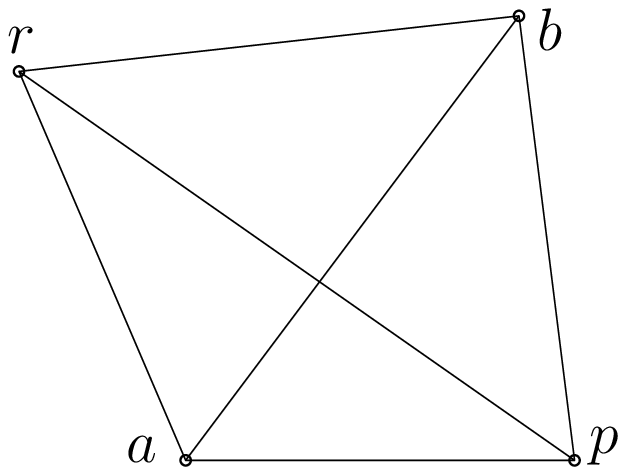} &
\includegraphics[height = 0.15\linewidth]{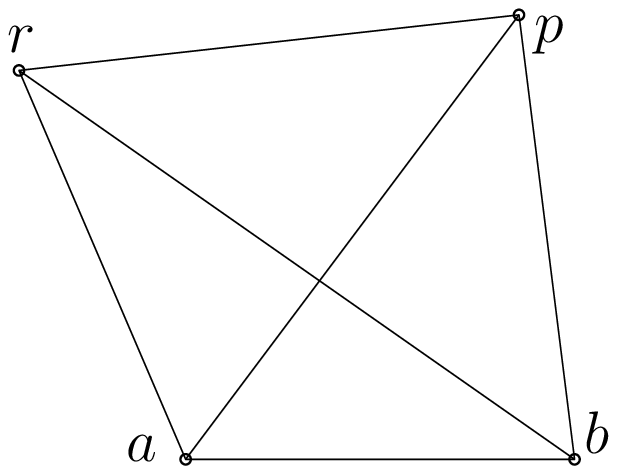} &
\includegraphics[height = 0.15\linewidth]{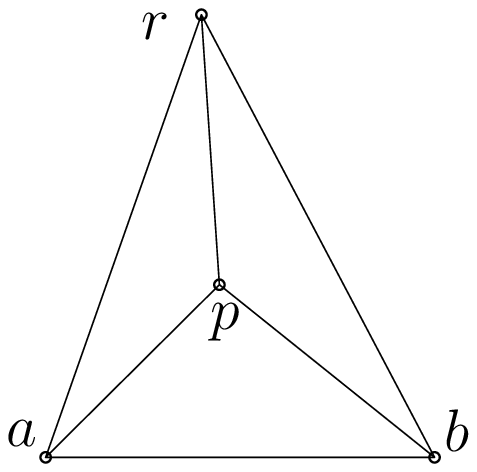} &
\includegraphics[height = 0.15\linewidth]{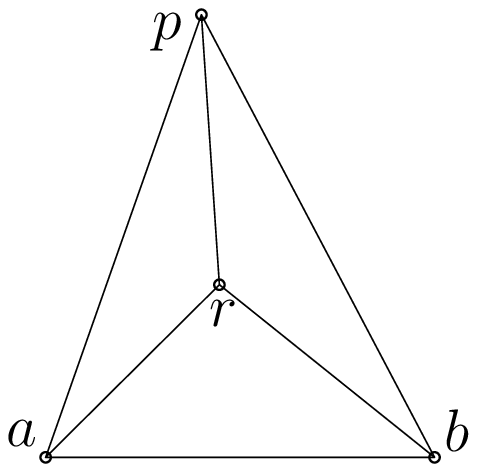} &
\includegraphics[height = 0.15\linewidth]{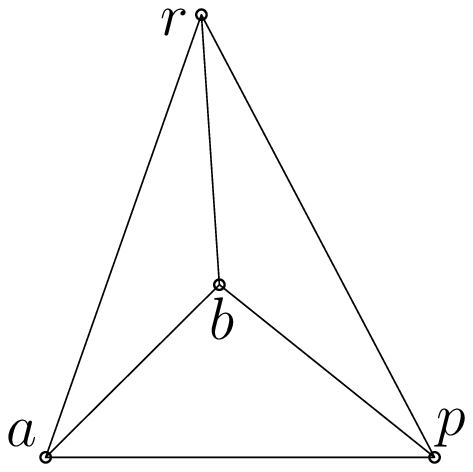} \\
(a) & (b) & (c) & (d) & (e)
\end{tabular}
\caption{The five non-isomorphic configurations needed to consider for a quadrilateral on four 
points in Lemma \ref{quad-lem}}\label{quad-fig}
\label{figure:crossingtest}
\end{figure}

Suppose the quadrilateral is convex. 
If $r$ and $p$ are on two opposite corners (see Figure \ref{quad-fig}(a)),
then $r$ and $p$ are on different sides. 
In this case $|rp|+|ab|> |ra|+|bp|$ and $|rp|+|ab|> |rb|+|ap|$. 
If $rp$ is one of its sides (see Figure \ref{quad-fig}(b)), 
then $|rp|+|ab|$ is not the largest of the above three pairs of sums.

Now suppose that the quadrilateral is concave. The only case in which $r$ and $p$ are on two sides of line
$ab$ is when one of $a$ or $b$ is inside the triangle obtained by the other three (see Figure \ref{quad-fig}(e)). 
In this case, the area of the largest triangle is the one that does not contain $a$ or $b$. Thus, if we
compute the square of the areas of the four triangle, we can detect this case too.
\end{proof}

Assume that $G'$ is a UDG and has an optimum clique partition 
of size $\alpha\leq \ell$.  The cliques fall in two categories:
small (having at most $2\alpha -2$ points), and large (having at least
$2\alpha -1$ points). We focus only on finding the large cliques
since it is easy to guess all the small cliques.  Suppose for each
pair $C_i,C_j \in \OPT$ of large cliques, we guess their respective
representatives, $c_i$ and $c_j$. Further, suppose that we also guess
a separating line $l_{ij}$ correctly which goes through
points $u_{ij}$ and $v_{ij}$. 
For a point $p$ that is adjacent to $c_i$ or $c_j$ 
we want to efficiently test
if $p$ is on the the same side of line $l_{ij}$ as $c_i$ (the
{\em positive} side), or on $c_j$'s side (the {\em negative} side),
using only edge-lengths.  Without loss of generality, let both 
$u_{ij}$ and $v_{ij}$ belong to clique $C_i$.
For every node $p$ different from the representatives:
\begin{itemize}
\item Suppose $p$ is adjacent to all of $c_i,u_{ij},v_{ij},c_j$. Observe
that we also have the edges $c_iu_{ij}$
and $c_iv_{ij}$. Given the edge-lengths of all the six edges among the
four vertices $c_i,u_{ij},v_{ij},p$ 
using Lemma \ref{quad-lem} we can
decide if in a realization of these four points, the line going through $u_{ij},v_{ij}$ separates the
two points $p$ and $c_i$ or not. If $p$ and $c_i$ are on the same side, we say $p$ is on the positive side
of $l_{ij}$ for $C_i$. Else, it is on the positive side of $l_{ij}$ for $C_j$.

\item Suppose $p$ is adjacent to $c_i$ (and also to $u_{ij}$ and $v_{ij}$)
but not to $c_j$. Given the edge-lengths of all the six edges among the
four vertices $c_i,u_{ij},v_{ij},p$ 
using Lemma \ref{quad-lem} we can
decide if in a realization of these four points, the line going through 
$u_{ij},v_{ij}$ separates the
two points $p$ and $c_i$ or not. If $p$ and $c_i$ are on the same side, 
we say $p$ is on the positive side
of $l_{ij}$ for $C_i$. Else, it is on the positive side of $l_{ij}$ for $C_j$.
\end{itemize}

For each $C_i$ and all the lines $l_{ij}$,
consider the set of nodes that are on the positive side of all these 
lines with respect to $C_i$; we place these nodes in $C_i$.
After obtaining the large and the small cliques, 
we obtain sets $C_1,\ldots,C_\alpha$. At the end we check if each
$C_i$ forms a clique and if their union covers all the points.
The number of guesses for representatives is $n^{O(\alpha)}$ 
and the number of guesses for the
separator lines is $n^{O(\alpha^2)}$. So there are a total of 
$n^{O(\alpha^2)}$ configurations that we consider.

Clearly, if $G'$ is a UDG then some set of
guesses is a correct one, allowing us to obtain an optimum clique partition.
If $G'$ is not a UDG, we may still find a clique partition of $G'$. 
However, if we fail to obtain a
clique partition in our search then the subgraph is a certificate that 
$G'$ is not a UDG.

\section{$(2+\eps)$-Approximation for Weighted Clique 
Partition using Adjacency}
\label{sec:wtd}
In this section we consider a generalization of the minimum clique partition on UDGs,
which we call {\em minimum weighted clique partition} (MWCP).
Given a node-weighted graph $G(V,E)$ with vertex weight $\wt{v}$,
the weight of a clique $C$ is defined as
the weight of the heaviest vertex in it.  For a clique partition 
$\mathcal{C} = \{C_1, C_2, \ldots, C_t\}$, the weight of $\mathcal C$ is defined
as sum of the weights of the cliques in $\mathcal C$, i.e.
 $\wt{\mathcal{C}} = \wt{\bigcup_{i=1}^t C_i} = \sum_{i=1}^t \wt{C_i}$. 
The problem is, given $G$ in standard form, say, as an adjacency matrix, 
construct a clique partition $\mathcal{C} = \{C_1, C_2, \ldots, C_t\}$ 
while minimizing $\wt{\mathcal{C}}$.  
The weighted version of the problem as it is defined above has also been
studied in different contexts. 
See \cite{FinkeJQS08,BecchettiKMSSV06,GijswijtJQ07}
for study of weighted clique-partition on 
interval graphs and  circular arc graphs.

Observe that MWCP distinguishes itself from MCP in two important ways: 
(i) The {\em separability} property which was crucially used 
earlier to devise a PTAS does not hold
in the weighted case, and (ii) the number of cliques in an optimal solution for
a UDG in a region of bounded radius is not bounded by the diameter 
of the region anymore,  
i.e. it is easy to construct examples of weighted UDGs in a bounded region 
where an optimal weighted clique partition contains an unbounded (in terms of
region diameter) number of cliques.  In addition, examples where two cliques in an optimal 
solution  are not separable, that is, their convex hulls overlap, is easy to
construct. (See the examples given in Figure~\ref{figure:distinction}.) 
To the best of our knowledge, MWCP has not been investigated before
on UDGs.  We, however, note that a simple modification to
the algorithm by \cite{PemmarajuPESA07} also
yields a factor-$8$ approximation to the weighted case, a generalization
which they do not consider. 
\begin{figure}[h]
\centering
\begin{tabular}{c@{\hspace{0.1\linewidth}}c}
\includegraphics[height = 0.25\linewidth]{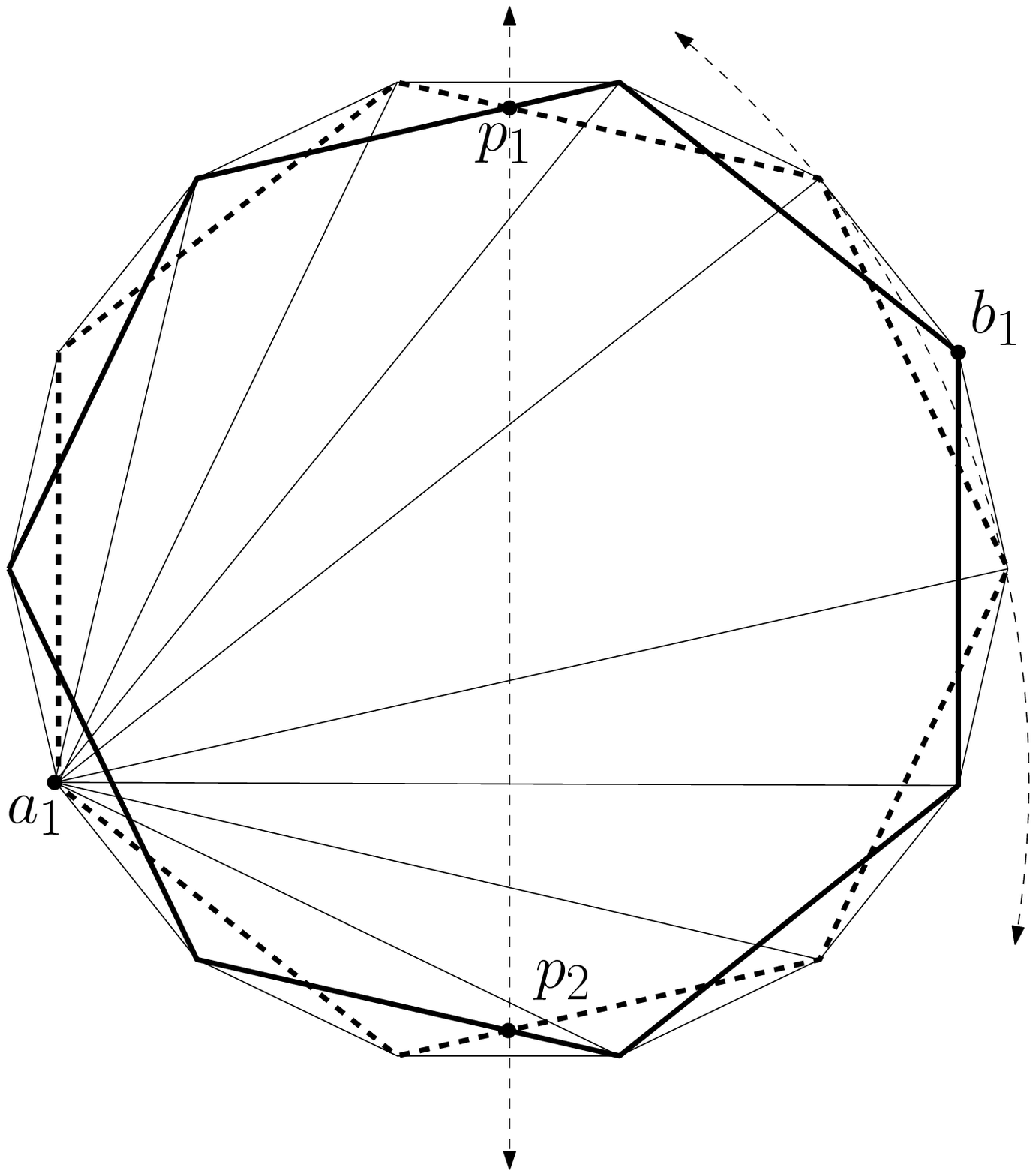} &
\includegraphics[height = 0.25\linewidth]{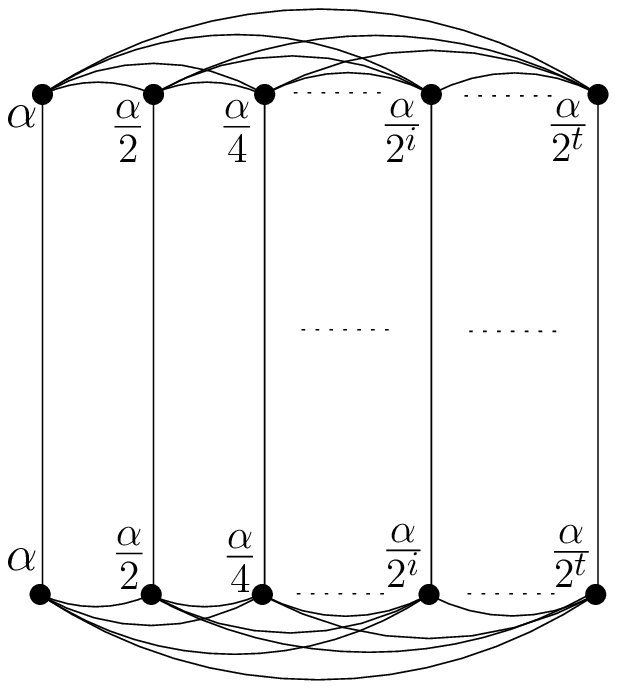} \\
(a) & (b)
\end{tabular}
\caption{(a) Two overlapping weighted cliques, $A = \{a_1, \ldots,a_k\}$ 
and $B = \{b_1, \ldots, b_k\}$ are shown, $a_i, b_i$ are independent for
all $i$. The heavy polygon has vertices weighted $k$ while the dashed ones 
are weighted $1$. $\opt = k + 1$ while any separable
partition must pay a cost of at least $2k$.
(b) A UDG which is a matching between two cliques for which $\OPT$ 
contains $t$ cliques.  The weight is less than $2 \cdotp \alpha$.
}
\label{figure:distinction}
\end{figure}
Here, we give an algorithm which runs in time $O(n^{\poly(1/\eps)})$
for a given $\eps > 0$ and computes a 
$(2+\eps)$-approximation to MWCP for UDGs expressed in standard form,
for example, as an adjacency matrix.  Our algorithm is weakly 
robust in that it either produces a clique partition or produces a 
polynomial-sized certificate proving that the input is not a UDG.
When the input is a UDG, the algorithm returns a clique partition
and it is guaranteed to be a
$(2+\eps)$-approximation; but if the input is not UDG there 
is no guarantee on the quality of the clique partition (if it computes one).

\begin{theorem}
Given a graph $G$ expressed in standard form, and $\eps > 0$, 
there is a polynomial time algorithm which either computes a clique 
partition of $G$ 
or gives a certificate that $G$ is not a UDG. If $G$ is a UDG, the
weight of the clique partition
computed is a $(2+\eps)$-approximation of the minimum weighted clique 
partition (but there is no guarantee on the weight of the clique partition 
if the input graph is not UDG).
\end{theorem}

Our algorithm will borrow some
ideas developed in Section~\ref{sec:2b} and in \cite{PemmarajuPESA07}.
The high level idea of the algorithm is as follows.
Similar to the algorithm in Section~\ref{sec:2b}, we first decompose the 
graph into bounded diameter regions
and show that if we can compute a $(2+\eps)$-approximate clique 
partition of each region then the union
of these clique partitions is within $(2+\eps)$ fraction of opt.  We 
will employ a similar ball growing technique (as in Section \ref{sec:2b})
that will give us bounded diameter regions. 
We then show that we can either compute a clique partition or 
give a certificate that the subgraph is not a UDG.
If the subgraph is a UDG, then our clique partition is within a factor
$(2+\eps)$ of the optimal. For the case of bounded diameter region, although the
optimum solution may have a large number of cliques, we can show that there
is a clique partition with small number of cliques whose cost is within
$(1+\eps)$-factor of the optimum solution. First we describe the main algorithm.
Then in Subsection \ref{sec:nearopt} we show that for each subgraph $B_r(v)$ (of bounded diameter)
there is a near optimal clique partition with $\tilde{O}(r^2)$ cliques. Then 
in Subsection \ref{sec2gamma} we show how to find such a near optimal clique partition.

Let us denote the weight of the optimum clique partition of $G$ by $\opt$.
As before, let $B_r(v) = \{u : d(u,v) \leq r\}$, called the ball of (unweighted) 
distance $r$ around $v$, 
be the set of vertices that are at most $r$ hops from $v$ in $G$. 
Our decomposition algorithm described below (see Algorithm 3) is similar
to Algorithm 2 and partitions the graph into bounded diameter subgraphs below.
The procedure \mbox{CP}, given a graph induced by the vertices of $B_r(v)$ 
and a parameter $\ell=\poly(r)$, runs in time $n^{O(\ell^2)})$ and either 
gives a certificate that $B_r(v)$ is not a
UDG or computes a clique partition of $B_r(v)$; this clique partition 
is within a factor $(2+\eps)$ of the optimum if $B_r(v)$ is a UDG. 
We only call this procedure for constant values of $r$.  In the
following, let $0 < \gamma \leq \frac{\sqrt{9+4 \eps}-3}{2}$ be 
a rational number. See Algorithm~3.

\begin{algorithm}[h]
	\caption{MinCP$(G,\gamma)$}
	\begin{algorithmic}[1]
		\STATE {$\mathcal{C} \leftarrow \emptyset$; 
						$\beta \leftarrow \lceil c_0
						\frac{1}{\gamma}\log\frac{1}{\gamma} \rceil$;
						$\ell \leftarrow c_1 \beta^2$.}

		\COMMENT {where $c_0$ is the constant in  
							Lemma~\ref{lemma:diamboundwtd},
							and $c_1$ is the constant in inequality~(\ref{c1-const-wtd}).}
		\WHILE {$V \neq \emptyset$}
			\STATE $v \leftarrow \arg\max_u\{\wt{u}\}$
			\STATE $r \leftarrow 0$

			\COMMENT {Let $C_{r}(v)$ denote a factor-$(2+\gamma)$ 
								partition of $B_r(v)$ computed by calling CP}
			\WHILE {$\wt{C_{r+2}(v)} > (1+\gamma) \cdotp \wt{C_{r}(v)}$}
				\STATE $r \leftarrow r + 1$
				\IF {($r > \beta$) or (CP$(B_r(v),\ell)$ returns 
						``not a UDG'')}
					\RETURN {``$G$ is not a UDG" and produce $B_r(v)$ as the
									 certificate}
				\ENDIF
			\ENDWHILE
			\STATE {$\mathcal{C} \leftarrow \mathcal{C} \cup C_{r+2}(v)$}
			\STATE $V \leftarrow V \setminus B_{r+2}(v)$
		\ENDWHILE
	\RETURN {$\mathcal{C}$ as our clique partition}
	\end{algorithmic}
\end{algorithm}

Let $k$ is the maximum number of iterations of the outer
``while-loop". The proof of the following Lemma is identical to the proof of Lemma~\ref{lemma:disjointballs}.

\begin{lemma}
Every two vertices $v\in B_{r^*_i}(v_i)$ and $u\in B_{r^*_j}(v_j)$ 
are non-adjacent.
\end{lemma}

The following lemma shows a lower-bound for $\opt$. 

\begin{lemma}
\label{lem:2approx}
$(2+\gamma) \cdotp \opt \geq 
				\wt{\displaystyle\bigcup_{i=1}^k C_{r^*_i}(v_i)}$
\end{lemma}
\begin{proof}
Note that $B_{r^*_i}(v_i)$ is obtained by constructing a BFS tree rooted at
vertex $v_i$ up to some depth $r^*_i$.  Since the algorithm
removes a super-set, $B_{r^*_i+2}(v_i)$, which has two more levels of
the BFS tree, using the previous lemma there is no edge between any 
two nodes $v\in B_{r^*_i}(v_i)$ and $u\in B_{r^*_j}(v_j)$ for any pair $i\not=j$.
So, no single clique in an optimum solution
can contain vertices from distinct $B_{r^*_i}(v_i)$ and $B_{r^*_j}(v_j)$. 
Consider the subset of cliques in an optimal clique partition of $G$ that
intersect $B_{r^*_i}(v_i)$ and call this subset $\OPT_i$.  The argument 
above shows that $\OPT_i$ is disjoint from $\OPT_j$.
Also, each $\OPT_i$ contains all the vertices in $B_{r^*_i}(v_i)$.
Since $C_{r^*_i}(v_i)$ is a factor-$(2+\gamma)$ approximation 
for $B_{r^*_i}(v_i)$,
$(2+\gamma) \cdotp \wt{\OPT_i} \geq \wt{C_{r^*_i}(v_i)}$. 
The lemma immediately follows by observing that $\OPT_i$ and $\OPT_j$ 
are disjoint.
\end{proof}

We can relate the cost of our clique partition
to $\opt$ as follows.

\begin{lemma}
\label{lem:2approxcost}
If $\wt{C_{r^*_i+2}(v_i)} \leq (1+\gamma) \wt{C_{r^*_i}(v_i)}$, then
$\wt{\displaystyle\bigcup_{i=1}^{k} C_{r^*_i+2}(v_i)} \leq
(2+\eps)\opt$.
\end{lemma}
\begin{proof}
$$\wt{\displaystyle\bigcup_{i=1}^{k} C_{r^*_i+2}(v_i)}= 
				\displaystyle\sum_{i=1}^{k} \wt{C_{r^*_i+2}(v_i)}
		\leq (1+\gamma) \cdotp \displaystyle\sum_{i=1}^{k}\wt{C_{r^*_i}(v_i)}
		\leq (2+\gamma)(1+\gamma) \cdotp \opt,
$$
where the last inequality uses Lemma~\ref{lem:2approx}.
\end{proof}

Next, we show that the inner ``while-loop" terminates in
$\tilde{O}(\frac{1}{\gamma})$, that is each $r^*_i$ is bounded by 
$\tilde{O}(\frac{1}{\gamma})$. This is similar to the proof of Lemma~\ref{lemma:diambound}.
Since the while loop terminates, $r^*_i$ exists and by definition 
of $r^*_i$, it must be the case that for all smaller values of
$r < r^*_i$, $\wt{C_r(v_i)} > (1+\gamma) \cdotp \wt{C_{r-2}(v_i)}$. 
Because the diameter of $B_r(v_i)$
is $O(r)$, if $B_r(v)$ is a UDG, there is a realization of it in which 
all the points fit into a $r\times r$ grid. Also, since $v_i$ is
a heaviest vertex in the (residual) graph, there is a clique partition
whose weight is at most $\alpha \cdotp \wt{v_i} \cdotp r^2$.
Therefore, $\wt{C_r(v_i)} < \alpha \cdotp \wt{v_i} \cdotp r^2$, 
for some constant $\alpha$.  So:
$$\alpha \cdotp \wt{v_i} \cdotp r^2 > \wt{C_r(v_i)} 
				> (1+\gamma) \cdotp \wt{C_{r-2}(v_i)}
				> \ldots 
				> (1+\gamma)^{\frac{r}{2}} \cdotp \wt{C_0(v_i)} 
				= \wt{v_i} \cdotp \left(\sqrt{1+\gamma}\right)^r,
$$ which implies  $\alpha \cdotp r^2 > \left(\sqrt{1+\gamma}\right)^r$, for the case that $r$ is even.
If $r$ is odd we obtain
$\alpha \cdotp r^2 > \left(\sqrt{1+\gamma}\right)^{r-1}$.
Thus, the following lemma easily follows:
\begin{lemma}\label{lemma:diamboundwtd}
There is a constant $c_0 > 0$ such that for each $i$: $r^*_i \leq 
c_0/\gamma \cdotp \log{1/\gamma}$.
\end{lemma}

In Subsection \ref{sec2gamma}, we show the algorithm CP that given $B_{r}(v)$
and an upper bound $\ell$ on $|C_r(v)|$, either
computes a clique partition (which is within a factor $2+\gamma$ of $\opt$ 
if $B_r(v)$ is a UDG) 
or detects that the graph is not UDG; the algorithm runs
in time $n^{O(\ell^2)}$. 
By the above arguments, if $B_r(v)$ is a unit disk graph then 
there is a constant $c_1 > 0$ such that:

\begin{equation}\label{c1-const-wtd}
|C_r(v)| = O({r_i^*}^2) \leq 
		c_1 \cdotp \frac{c_0^2}{\gamma^2}\log^2\frac{1}{\gamma}
\end{equation}

We can 
set $\ell = \lceil c_1 \frac{c_0^2}{\gamma^2}\log^2\frac{1}{\gamma} \rceil$
for any invocation of OPT-CP as an upper bound,
where $c_1$ is the constant in $O({r^*_i}^2)$.
So, the running time of the algorithm is 
$n^{\tilde{O}(1/\eps^4)}$.

\subsection{Existence of a Small Clique Partition of $B_r(v)$ having 
Near-optimal Weight}\label{sec:nearopt}
Unlike the unweighted case, an optimal weighted clique
partition in a small region may contain a large number of cliques.
Yet, there exists a partition whose weight is within a factor
$(1+\frac{\gamma}{2})$ of the minimum weight which contains few cliques
(where by ``few'' we mean $\ell$ as in Algorithm 3).  The existence of 
a light and small partition allows us to enumerate them 
in the same manner in the algorithm of subsection \ref{secbrv}, 
yielding a $(2+\gamma)$-approximation for the problem
instance in a ball of small radius. 
In the following, let $r \in \tilde{O}(\frac{1}{\gamma})$;
we focus on the subproblem that lies in some $B_r(v)$. 
Recall that any ball of radius $r$ can be partitioned into
$O(r^2)$ cliques (Lemma~\ref{lemma:fewcliques}).
We begin with a simple lemma which states
that for any clique partition $\mathcal{C}$, if the set of vertices can
be  be covered by another clique partition ${\mathcal C}'$ containing 
$x$ cliques then the sum of the weights of the $x$
cliques in ${\mathcal C}'$ is not significantly more than the 
weight of the heaviest
clique in $\mathcal{C}$.

\begin{lemma}
For any collection of disjoint cliques 
$\mathcal{C} = \{C_1, C_2, \ldots, C_t\}$ having
weights such that $\wt{C_1} \geq \wt{C_2} \geq \ldots \geq \wt{C_t}$ suppose
the vertices of $\mathcal C$ can be partitioned into $x$ cliques  
$\mathcal{C}^\prime = \{C_1^\prime,C_2^\prime, \ldots, C_{x}^\prime\}$.
Then $\wt{\mathcal{C}^\prime} = \wt{\bigcup_{l=1}^{x} C_l^\prime}
= \sum_{l=1}^{x} \wt{C_l^\prime} \leq x \cdotp \wt{C_1}$
\label{lemma:twoclique}
\end{lemma}
\begin{proof}
Without loss of generality, let 
$\wt{C_1^\prime} \geq \wt{C_2^\prime} \geq \ldots \geq \wt{C_{x}^\prime}$.
Since $\mathcal{C}^\prime$ partitions vertices in $\mathcal{C}$,
$\wt{C_1^\prime} = \wt{C_1}$.  Since $|\mathcal{C}^\prime| = x$,
$\wt{\mathcal{C}^\prime} = \sum_{l=1}^{x} \wt{C_l^\prime} \leq
x \cdotp \wt{C_1^\prime} = x \cdotp \wt{C_1}$.
\end{proof}

In an optimal partition of a ball of
radius $r$, the sum of 
the weights of the lighter cliques is not significantly more
than its weight.

\begin{lemma}
Let $\mathcal{C} = \{C_1, C_2, \ldots, C_t\}$ be an optimal clique partition 
and let $\wt{C_1} \geq \wt{C_2} \geq \ldots \geq \wt{C_t}$. Suppose there
is another clique partition $\mathcal C'=\{C'_1,\ldots,C'_x\}$ of the vertices of $\mathcal C$.
Then, for every $1\leq i<t$: $(x-1) \cdotp \wt{C_i} \geq \sum_{l=i+1}^t \wt{C_l}$. 
\label{lemma:heavycost}
\end{lemma}
\begin{proof}
By way of contradiction, suppose there exists an index $1\leq j<t$
such that $(x-1) \cdotp \wt{C_j} < \sum_{l=j+1}^t \wt{C_l}$.  
Because $\bigcup_{l=1}^t C_l$ can be covered by $\mathcal{C}^\prime$, so can
$\bigcup_{l=j}^t C_l$.  Let the $2\leq x' \leq x$ be the smallest index such that
$\mathcal{C}_j^\prime = \{C_{1}^\prime, C_{2}^\prime, \ldots, 
C_{x'}^\prime\}$ covers $\bigcup_{l=j}^t C_l$.
On the other hand, $(x'-1) \cdotp \wt{C_j} \leq (x-1)\cdotp \wt{C_j} < \sum_{l=j+1}^t \wt{C_l}$,
which implies 
\begin{equation}\label{eqn11}
\opt  = \sum_{l=1}^t \wt{C_l} > \sum_{l=1}^{j} \wt{C_l} + (x'-1) \cdotp \wt{C_j}
= \sum_{l=1}^{j-1} \wt{C_l} +x'\wt{C_j}.
\end{equation}
By Lemma~\ref{lemma:twoclique}, 
$\wt{\mathcal{C}_j^\prime} \leq x' \cdotp \wt{C_j}$. This, 
combined with inequality (\ref{eqn11}) implies
$\opt > \sum_{l=1}^{j-1} \wt{C_l} + \sum_{l=1}^{x'} \wt{C_{l}^\prime}$.
Therefore the cliques in ${\mathcal C''}=\{C_1, C_2, \ldots, C_{j-1}, C_1^\prime, C_2^\prime, \ldots,
C_{x'}^\prime\}$ cover all the nodes of cliques in $\mathcal C$ and has
cost smaller than $\opt$. If a vertex belongs to two or more cliques in $\mathcal C''$ 
we remove it from all but one of them to obtain a clique partition with cost no more than cost of $\mathcal C''$
which is smaller than $\opt$. This completes the proof.
\end{proof}

We now are ready to prove the main result of 
this section which
states that for any optimal weighted clique partition of a ball
of radius $r$, there exists another clique partition whose weight
is arbitrarily close to the weight of the optimal partition, but has 
$O(r^2)$ cliques in it.  Since the radius of the ball within which
the subproblem lies is small, $r \in \tilde{O}(\frac{1}{\gamma})$, this 
means that if we were to enumerate all the clique partitions of the
subproblem up to $O(r^2)$, we will see one whose weight is 
arbitrarily close to the weight of an optimal clique.
Choosing a lightest one from 
amongst all such cliques guarantees that we will choose a one
whose weight is arbitrarily close to the optimal weight.

\begin{lemma}\label{lemgamma}
Let $\gamma > 0$ and $r \in \tilde{O}(1/\gamma)$ be two constants.
Let $\mathcal{C} = \{C_1, C_2, \ldots, C_t\}$ be an optimal weighted
clique partition of $B_r(v)$ and let ${\mathcal C'}=\{C'_1,\ldots,C'_x\}$
be another clique partition of vertices of $\mathcal C$ with $x \in O(r^2)$.  
Let $\wt{C_1} \geq \wt{C_2} \geq \ldots \geq \wt{C_t}$.
Then, there is a partition of vertices of $\mathcal C$ into
at most $j+x$ cliques for some constant $j = j(\gamma)$, with cost at most 
$(1+\frac{\gamma}{2})\opt$.
\label{lemma:boundedcliques}
\end{lemma}
\begin{proof}
Without loss of generality, we assume that both $x$ and $t$ are at least two 
(as if $B_r(v)$ is a clique we are done). Consider an arbitrary value of $j\leq t$. 
Since $\bigcup_{l=1}^t C_l$ can be covered by $x$ cliques in $\mathcal C'$,
there is an index $x'$ ($2\leq x'\leq x$) such that $\bigcup_{l=j}^t C_l$ can be covered 
by $\mathcal{C}_j^\prime = \{C_1^\prime, C_2^\prime, \ldots, C_{x'}^\prime\}$.
By applying Lemma~\ref{lemma:heavycost} repeatedly:

\[
	\opt  \geq  \sum_{l=1}^j \wt{C_l} 
			  \geq  \frac{1}{x'-1} \left(
			 				\sum_{l=2}^j \wt{C_l}\right) + \sum_{l=2}^j \wt{C_l} 
			  \geq  \ldots 
			  \geq  \left(\frac{x'}{x'-1}\right)^{j-1} \cdotp \wt{C_j} 
\]
\begin{equation}\label{eqncj}
	\Rightarrow  \opt \frac{(x'-1)^{j-1}}{x'^{j-2}} \geq  x'\cdotp \wt{C_j}
\end{equation}
Using inequality (\ref{eqncj}):

\begin{equation}\label{eqn12}
	\opt + \frac{\opt\cdotp (x'-1)^{j-1}}{x'^{j-2}}  \geq  
					\sum_{l=1}^{j-1} \wt{C_l} + x' \cdotp \wt{C_j} 
				 \geq  \sum_{l=1}^{j-1} \wt{C_l} + 
						\sum_{l=1}^{x'}\wt{C_l^\prime},
\end{equation}
where the second inequality follows by applying Lemma~\ref{lemma:twoclique}.
Let ${\mathcal C''}=\{C_1,\ldots,C_{j-1},C'_1,\ldots,C'_{x'}\}$. Thus, the cliques
in $\mathcal C''$ cover all the vertices of $\mathcal C$ and has total cost at most
$\left( 1 + \frac{(x'-1)^{j-1}}{x'^{j-2}}\right) \opt$ by inequality (\ref{eqn12}).
If a vertex belongs to two or more cliques in ${\mathcal C''}$ 
we remove it from all but one of them
arbitrarily to obtain a clique partition of size $j-1+x'$ and whose total cost is upper bounded by
$\left( 1 + \frac{(x'-1)^{j-1}}{x'^{j-2}}\right) \opt$.
Note that, $\frac{(x'-1)^{j-1}}{x'^{j-2}} = (x'-1)\left(\frac{x'-1}{x'}\right)^{j-2}$
and $0< \frac{x'-1}{x'} < 1$ (because $x'\geq 2$). Since $r \in \tilde{O}(\frac{1}{\gamma})$
and $x' \leq x \in O(r^2)$, for an appropriate choice of $j = j(\gamma)$,
$(x'-1)\left(\frac{x'-1}{x'}\right)^{j-2}<\gamma/2$. Thus we obtain a clique partition with $j+x-1$
cliques and cost at most $(1+\gamma/2)\cdotp \opt$.
This proves the lemma. 
\end{proof}

\subsection{$(2+\gamma)$-Approximation for MWCP in $B_r(v)$}\label{sec2gamma}
Finally, we show how to compute a $(2+\gamma)$-approximate MWCP of the
graph $B_r(v)$ for any given $\gamma$.
For an edge ordering $L = (e_1, e_2, \ldots, e_m)$ of a graph $G$ with $m$ edges, 
let $G_L[i]$ denote the edge induced subgraph with edge-set
$\{e_i, e_{i+1}, \ldots, e_m\}$. For each $e_i$, let
$N_L[i]$ denote the common neighborhood of the end-points of $e_i$
in $G_L[i]$.   An edge ordering $L = (e_1, e_2, \ldots, e_m)$ is a CNEEO
if for every $e_i$ in $L$, $N_L(i)$ induces a co-bipartite 
graph in $G$. It is known \cite{RaghavanS03} that every UDG graph
admits a co-bipartite edge elimination ordering (CNEEO).
In the following, let $G_v$ denote $B_r(v)$.  We state a lemma
of \cite{PemmarajuPESA07}.

\begin{lemma}{\cite{PemmarajuPESA07}}
\label{lemma:catchclique}
Let $C$ be a clique in $G_v$, and let $L$ be a CNEEO of $G_v$.
Then, there is an $i$, $1 \leq i \leq m$, such that
$N_L[i]$ contains $C$.
\end{lemma}

Assume that $G_v$ can be partitioned into $\alpha\leq \ell=\tilde{O}(1/\gamma^2)$ 
cliques, $\mathcal{O} = \{O_1, O_2, \ldots, O_{\alpha}\}$, such that
$\wt{\mathcal{O}} \leq (1+\frac{\gamma}{2})\cdotp\wt{\OPT_v}$, where $\OPT_v$
is an optimal weighted clique partition of $G_v$. Note that by Lemma \ref{lemgamma}
this is true for subgraph $B_r(v)$.
Suppose that we are given the upper bound $\ell$; we will try 
all possible values of $\alpha$.
Without loss of generality, let $\wt{O_1} \geq \wt{O_2} \geq \ldots \geq
\wt{O_{\alpha}}$.
Observe that, without loss of generality, 
we can assume $O_i$ is a maximal clique in $\bigcup_{j=i}^{\alpha} O_j$.  
The implication of the above lemma
is that even though we do not know $O_1$, hence we do not know
$\mathcal{O}$, we do know that for
every CNEEO $L$ of $G_v$, there is an $e_i$ such that $N_L[i]$ can
be partitioned into at most two cliques that fully cover $O_1$. 
Since $O_1$ is a heaviest clique, the two cliques that cover the
subgraph $N_L[i]$ pay a cost of at most $2 \cdotp \wt{O_1}$.
This suggests an algorithm that guesses an edge sequence
$(f_1, f_2, \ldots, f_{\alpha})$ of $G_v$.  Then, the algorithm computes
$L$, a CNEEO of $G_v$.  The algorithm's first guess is ``good" if
$f_1$ is an edge in $O_1$ that occurs first in $L$.  Suppose that
this is the case and suppose that $f_1$ has rank $i$ in $L$.  Then,
$O_1$ is contained in $N_L[i]$, and we cover $N_L[i]$ with at most
two cliques. Call these $C_1^\prime$ and $C_1^{\prime\prime}$ and
$\wt{C_1^\prime} + \wt{C_1^{\prime\prime}} \leq 2 \cdotp \wt{O_1}$.  
So, when we remove $N_L[i]$ from
$G_v$, we get a UDG which can be partitioned into at most $\alpha -1$
cliques, namely, $\mathcal{O}^\prime = \{O_2, \ldots, O_{\alpha}\}$. 
We then again construct a CNEEO, $L^\prime$, of 
$G_v^\prime = G_v \setminus N_L[i]$.  Just like before, 
our guess $f_2$ is ``good" if $f_2$ is an edge in $O_2$ and occurs
first in $L^\prime$.  
Let $i^\prime$ be the rank of $f_2$ in $L^\prime$, we see that
$N_{L^\prime}[i^\prime]$ fully contains $O_2$, and we again cover it with
at most $2$ cliques.  Next, delete $N_{L^\prime}[i^\prime]$ from
$G_v^\prime$ to get a graph which can be partitioned into $\alpha -2$
cliques, and so on.  See Algorithm~4 for details.

\begin{algorithm}[h]
	\caption{CP$(G_v,\ell)$}
	\begin{algorithmic}[1]
		\STATE $\mathcal{C} \leftarrow V$; min $\leftarrow \wt{\mathcal{C}}$; 
		\FORALL {$\alpha\leq \ell$}
		\FORALL {$\alpha$-edge sequence $(f_1, f_2, \ldots, f_{\alpha})$ of $G_v$}
			\STATE {$G_0 \leftarrow G_v$}
			\FOR {$j=1$ to $\alpha$}
				\STATE {Compute a CNEEO $L$ of $G_{j-1}$}
				\STATE {$i \leftarrow$ rank of $f_j$ in $L$}
				\STATE {Partition $N_L[i]$ into two cliques $C_j^\prime$ and
								$C_j^{\prime\prime}$}
				\STATE {$G_j \leftarrow G_{j-1} \setminus N_L[i]$}
			\ENDFOR
			\IF {$G_{\alpha} = \emptyset$ and $\wt{\bigcup_{j=1}^{\alpha}
				\{C_j^\prime, C_j^{\prime\prime}\}} < \mbox{min}$}
				\STATE $\mathcal{C} \leftarrow 
									\bigcup_{j=1}^{\alpha} \{C_j^\prime,
									C_j^{\prime\prime}\}$; 
									min $\leftarrow \wt{\mathcal{C}}$;
			\ENDIF
		\ENDFOR
		\ENDFOR
		\RETURN $\mathcal{C}$
	\end{algorithmic}
\end{algorithm}

Note that while Lemma~\ref{lemma:catchclique} allows us to cover
any clique with at most $2$ cliques, it does not find
the clique.  In the algorithm, note that if at any point,
the algorithm is unable to construct a CNEEO, we can declare that
the graph $G_v$ is not a UDG.  Also, if for all invocations of
the algorithm by an external algorithm that guesses the value of $\opt$
we are unable to find a clique partition, then again we can declare
that $G_v$ is not a UDG.

\section{$O(\log^* n)$-round Distributed PTAS for UDGs with Edge-Lengths}
\label{sec:distrib}

In this section,
we give details of a distributed PTAS for MCP which runs in
$O(\frac{\log^* n}{\eps^{O(1)}})$ rounds of distributed computation under the
$\mathcal{LOCAL}$ model of computation \cite{Peleg}. The model of
computation that we employ assumes a {\em synchronous} system
where {\em communication} between neighboring nodes takes place
in synchronous rounds using messages of unbounded size \cite{Peleg}.
So, in a single round of communication, any node acquires the
subgraph (information pertaining to the set of nodes, edges, the
states of local variables, etc.) within its immediate neighborhood.
So, after $k$ rounds of
communication, any node acquires complete knowledge about its
$k$-neighborhood.

Observe that in Algorithm MinCP2, the radius $r$ of any ball
$B_v(r)$ is bounded above by $\tilde{O}(1/\eps)$, while the center,
$v$, is an arbitrary vertex. Since the radius of
any ball is ``small", the maximum number of rounds of distributed 
computation that the sequential algorithm needs before terminating
the ``while-loop" is also ``small". Therefore, for any
pair of balls $B_u(r_i)$ and $B_v(r_j)$, such that 
$d(u,v) \in \omega(1/\eps)$, one should be able to run part of the 
sequential algorithm in parallel, as they surely are independent
of each other.  
We borrow some ideas from \cite{KuhnNMW05} and find regions that
are far apart
such that we can run the sequential algorithm in those regions in parallel.
See Algorithm~5 for details.

\begin{algorithm*}[h]
	\caption{Distr-MCP-UDG$(G,\eps)$}
	\begin{algorithmic}[1]
		\STATE {$\beta \leftarrow \lceil c_0
						\frac{1}{\eps}\log\frac{1}{\eps} \rceil$;
						$\ell \leftarrow c_1 \beta^2$; all vertices are unmarked.}

		\COMMENT {$c_0$ is the constant in Lemma~\ref{lemma:diambound} 
							and $c_1$ is the constant inequality (\ref{c1-const}).}
		\STATE Construct a maximal subset, $V_c \subset V$, such that
					 for any pair $u,v \in V_c$,
					 $d(u,v) > \beta$. Construct a graph $G_c=(V_c,E_c)$, where
					 $E_c=\{\{u,v\} : u,v \in V_c, d_G(u,v) \leq 4\beta\}$. We
					 call $V_c$, the set of {\em leaders}.
		\STATE Proper color $G_c$ using $\Delta(G_c)+1$ colors, where
					 $\Delta(G_c)$ is the maximum degree of $G_c$.
		\STATE Every $v \in V \setminus V_c$, ``assigns" itself to a nearest
					 leader	$u \in V_c$, with ties broken arbitrarily,
					 and colors itself the same color as the leader.
		\FOR 	 {$i = 1$ to $\Delta(G_c)+1$}
			\STATE 	For each leader $j$ with color $i$ let $G^j_i$ be the subgraph induced by the
				vertices assigned to leader $j$.
			\FORALL	{$G^j_i$ {\bf in parallel}} 
				\STATE Consider a fixed ordering on the unmarked vertices of $G^j_i$; 
				\STATE Run the sequential ball growing algorithm on the next (in this ordering) 
				unmarked vertex $v\in G^j_i$, we compute $B_r(v)$;
  			        Note that $B_r(v)$ might contain vertices of different colors (from outside $G^j_i$).
				\STATE Compute (using the sequential algorithm) the optimal clique-partition of $B_r(v)$
				and ``mark'' all those vertices 
				

			\ENDFOR
		\ENDFOR
	\end{algorithmic}
\end{algorithm*}

It should be pointed out that adapting the algorithm of \cite{KuhnNMW05} for
{\em maximum independent set} and {\em minimum dominating set} to our setting
is not trivial. The reason is that MCP is a {\em partition} of the entire
vertex set and partitioning just a subset well enough will not do.
Specifically \cite{KuhnNMW05} chooses a subset of vertices upon
which their ball-growing algorithm is run; it suffices for their purposes to
dispense with the remaining subset of vertices that were not picked by their
ball-growing algorithm.  If we had followed a similar scheme then we would
surely get a good clique partition on a subset of vertices; however, it is
unclear as to how to obtain a good partition of the remaining subset in terms
of the optimal size for the original problem instance over the entire vertex
set. As a means to circumvent this issue, we first construct a ``crude"
partition of the vertex set, instead of just a subset of vertices as done in
\cite{KuhnNMW05}.  


\subsection{Analysis}
We now show that the algorithm constructs a $(1+\eps)$-approximation
to MCP on UDGs given only rational edge-lengths in $O(\log^* n)$
rounds of distributed computation under the $\mathcal{LOCAL}$ model;
we first show correctness of the algorithm, followed by bounding
the number of communication rounds. 

\paragraph{Correctness:}
We prove that our algorithm is correct by showing
that any execution of Distr-MCP-UDG can be turned
into a sequential execution of MinCP2. 
As stated, every vertex has the same color as its leader; 
let a leader vertex be its own leader.
We first show that the distance of every vertex to its leader is small.

\begin{lemma}
For any vertex $v \notin V_c$, there is a vertex $u \in V_c$ such
that $d_G(u,v) \leq \beta$.
\label{lemma:colorpartition}
\end{lemma}
\begin{proof}
Suppose not. So there is a 
$v$ whose distance to every $u \in V_c$ is more than 
$\beta$. But then $V_c^\prime = V_c \cup \{v\}$
has the property that for all $x,y \in V_c^\prime$, $d_G(x,y) > 
\beta$, contradicting the maximality of $V_c$.
\end{proof}

Next we show that for any pair of vertices $u,v$ of the same color but with
different leaders, the minimum
distance between them is large enough so that a pair of balls 
of radius at most $\beta$ over them will be disjoint,
where $\beta$ is defined in MinCP2 and Distr-MCP-UDG.

\begin{lemma}
Consider two leaders $x,y$ of the same color, say $i$, and any two vertices
$u\in G^x_i$ and $v\in G^y_i$ (note that we might have $u=x$ or $v=y$).
Then for all values of $r$ considered in the ball growing algorithm, $B_r(u)$ and $B_r(v)$ are disjoint.
\label{lemma:mindist}
\end{lemma}
\begin{proof}
Since $x$ and $y$ have the same color $d(x,y)>4\beta$.
By Lemma~\ref{lemma:colorpartition}, any
vertex in either of $G^x_i$ or $G^y_i$ is at a distance of at most
$\beta$ from the respective leader; so, $d_G(u,v) > 2\beta$.
The lemma follows easily by noting the fact that $r\leq \beta$ in the
ball growing algorithm.
\end{proof}

We are now ready to prove the correctness of
Distr-MCP-UDG by showing an equivalence between
any execution of it to some execution of MinCP2.

\begin{lemma}
Any execution of Distr-MCP-UDG from 
``Step 5" to ``Step 10" can be converted to a valid 
execution of MinCP2.
\label{lemma:equivalence}
\end{lemma}
\begin{proof}
Consider an arbitrary execution of Distr-MCP-UDG.
Suppose that $V_1, V_2, V_3,\ldots$ is a sequence of disjoint 
sets of the vertices of $V$ such that we run the ball growing algorithm
in parallel (during Distr-MCP-UDG) on vertices of $V_1$ 
(and thus we compute an optimal clique partition on each vertex of 
$V_1$ in parallel) then we do this for vertices in $V_2$, and so on. 
Note that the vertices in $V_i$ all have the same color and each has a 
different leader.  Consider an arbitrary ordering $\pi_i$ 
of the vertices in each $V_i$ and suppose that we run MinCP2 
algorithm on vertices of $V_1$ based on ordering $\pi_1$, 
then on vertices of $V_2$ based on ordering $\pi_2$, and so on. 
Since the vertices in each $V_i$ have distinct leaders, 
by Lemma \ref{lemma:mindist}, the balls grown around them are disjoint. It 
should be easy to see that the balls grown by algorithm MinCP2 is
exactly the same as the ones computed by Distr-MCP-UDG.
\end{proof}

The following result follows immediately as a corollary to
Lemma~\ref{lemma:equivalence}.

\begin{cor}
Given an $\eps>0$, Distr-MCP-UDG constructs a clique partition
of the input graph $G$ with associated edge-lengths, 
or produces a certificate that $G$ is not a UDG.  If $G$ is a UDG
then the size of the partition is within $(1+\eps)$ of the
optimum clique partition.
\end{cor}

\paragraph{Running Time:}
We now show that the algorithm runs in $O(\frac{\log^* n}{\eps^{O(1)}})$ 
distributed rounds under the $\mathcal{LOCAL}$ model of computation.

\begin{lemma}
``Step 2" requires $O(\beta \cdotp \log^* n)$ rounds of communication. 
\end{lemma}
\begin{proof}
Observe that the result of ``Step 2" is identical to constructing
a {\em maximal independent set} (MIS) in $G^\beta$.  Note that $G^\beta$
is also a UDG where the new unit is $\beta$.  As a result, $G^\beta$
is a subclass of {\em growth-bounded graphs} \cite{KuhnMNW05}
where all the distances are scaled by $\beta$; computation
of MIS on $G^\beta$ takes $O(\beta \cdotp \log^* n)$ rounds 
\cite{SchneiderW08} while the construction
of $G^\beta$ takes $\beta$ rounds.  Hence, the number
of rounds needed by ``Step 2" can be bounded by 
$O(\beta \cdotp \log^* n)$.
\end{proof}

It is easy to see that constructing $G_c$ requires at most $4\beta$
communication rounds.  Next, we show that the maximum degree
of $G_c$, $\Delta(G_c)$ is bounded by a constant. 

\begin{lemma}
$\Delta(G_c) \in O(1)$ 
\label{lemma:boundeddegree}
\end{lemma}
\begin{proof}
Let $v$ be a vertex of $G_c$ having maximum degree.  Note that all
its neighboring vertices in $G_c$ lie in a disk of radius at most
$4\beta$.  Also note that due to ``Step 2" the minimum distance between
any pair of vertices in $G_c$ is more than $\beta$.  As a result,
any disk of diameter $\beta$ contains at most $1$ vertex of $G_c$. 
Using standard packing arguments of the underlying space, a crude upper bound
on the number of vertices of $G_c$ in a disk of radius at most $4\beta$
is $256$ vertices; this also upper bounds the degree of $v$. 
\end{proof}

Next, we bound the number of rounds needed for ``Step 3" 

\begin{lemma}
``Step 3" requires $O(\beta \cdotp \log^* n)$ rounds of communication.
\label{lemma:coloring}
\end{lemma}
\begin{proof}
For graphs whose maximum degree is $\Delta$, a $\Delta + 1$ proper
coloring requires $O(\Delta + \log^* n)$ rounds 
\cite{KuhnW06}.  Since $\Delta(G_c) \in O(1)$ 
(Lemma~\ref{lemma:boundeddegree}), and the fact that distances in $G_c$
are scaled by a factor of $4\beta$ as compared to the distances
in $G$, a $\Delta(G_c) + 1$ proper coloring of $G_c$ 
can be obtained in $O(\beta \cdotp \log^* n)$ rounds.
\end{proof}

``Step 4" requires at most $\beta$ rounds
of communication; according to Lemma~\ref{lemma:colorpartition}, 
for every $v \notin V_c$, there is some $u \in V_c$ that is at a distance
at most $\beta$ from it. The identity and color of such a vertex
can be obtained in $\beta$ rounds.  We can now bound the number
of rounds that Distr-MCP-UDG requires. First,
note that for any iteration, $i$, of ``Step 7", only knowledge of
a subgraph up to radius $\beta$ is required, and any node can
obtain knowledge of the subgraph up to radius $\beta$ from it in
$\beta$ rounds of communication.  So, for any vertex in $G^j_i$ 
obtains knowledge about the ``marked/unmarked" status of all the vertices in $G^j_i$
in $\beta$ rounds of communication.
Since the diameter of each $G^j_i$ is at most $2\beta$, the number of
balls to grow in ``Step 9.'' is at most $O(\beta^2)$. Therefore:

\begin{theorem}
Distr-MCP-UDG requires 
$O(\beta \cdotp \log^* n)$ rounds of communication under the
$\mathcal{LOCAL}$ model of computation.
\end{theorem}


\section{Concluding Remarks}
Recall that the weakest assumption that we needed to obtain a PTAS for
unweighted clique partition problem was
that all the edge lengths are given.  This
information was crucially used in obtaining a {\em robust} PTAS.
In the case of weighted clique partition, we gave a $(2+\eps)$-approximation algorithm
without the use of edge-lengths (using only the adjacency information). It will be
interesting to see if a PTAS exists for the unweighted case but with
reliance only on adjacency. 

It is also unclear if a PTAS is possible
even with the use of geometry in the weighted case.  Recall that
the PTAS given in Sections \ref{sec:2b} crucially uses the idea of
separability of an optimal clique partition.  However,
in the weighted case, even though a near optimal clique partition
in a small region has few cliques, there are examples where any
separable partition pays a cost at least factor-$2$ to that of 
a near optimal partition.  We give an example in
Figure~\ref{figure:distinction}(a). 
In the example shown in Figure~\ref{figure:distinction}(a) two cliques of
optimal weight are shown: one of them, $A$, whose vertices are the 
vertices of the $k$-gon shown in dashed-heavy lines, and the other, $B$,
whose vertices are the vertices of the $k$-gon shown in solid-heavy
lines.  The example is that for $k=7$. The vertices of $A$ are labeled
$a_1, a_2, \ldots, a_k$ in a counter-clockwise fashion.  The vertices
of $B$ are labeled such that $b_i$ is diametrically opposite to $a_i$.
The distance between $a_i$ and $b_i$ is more than $1$ while the distance
between $a_i$ and $b_j$, $i \neq j$ is at most $1$.  So, there is an
edge between $a_i$ to every $a_l$ and to every $b_j, j \neq i$.  This
is also the case for $b_i$.  In the figure, the edges incident to
$a_1$ are shown by solid-light lines.  Also, the dashed arc shows
part of the unit disk boundary that is centered at $a_1$ -- note that
it does not include $b_1$.  Let the weights of vertices in $A$ be
$k$ and the weights of vertices in $B$ be $1$.  Clearly, 
$\opt \leq k+1$.  However, any separable clique partition pays a
cost of at least $2k$: if vertices in $A$ must all belong to a common
clique, then every vertex in $B$ must belong to a distinct clique in
a separable clique partition.  Also, note that as-per separability, 
a line going through $\{p_1, p_2\}$ separates
two cliques having weight $2k$ also.

Note that our results only apply in the Euclidean plane; they do not
generalize. In particular, Capoyleas et al. \cite{CapoyleasRW91} give
an ``unseparable" instance in $\mathbb{R}^3$. Our result in the weighted
case also is restricted to the plane; the concept of {\em co-bipartite
neighborhood edge elimination ordering} (CNEEO) does not generalize to
$\mathbb{R}^3$.

\section*{Acknowledgments}
We thank Sriram Pemmaraju, Lorna Stewart, and Zoya
Svitkina for helpful discussions. Our thanks to 
an anonymous source for pointing out the result of
Capoyleas et al. \cite{CapoyleasRW91}.

\bibliography{udgmcp}
\end{document}